\documentclass[submission,copyright,creativecommons]{eptcs}
\providecommand{\event}{TERMGRAPH 2018} 
\usepackage{breakurl}             
\usepackage{underscore}           

\usepackage{amscd}
\usepackage{amsmath}
\usepackage{amssymb}
\usepackage{amsthm}
\usepackage{proof}
\usepackage[all]{xy}
\usepackage{graphicx}

\newtheoremstyle{break}
 {10pt}
 {10pt}
 {}
 {}
 {\bfseries}
 {.}
 { }
 {}
\theoremstyle{break}
\newtheorem{defn}{Definition}[section]
\newtheorem{thm}[defn]{Theorem}
\newtheorem{lem}[defn]{Lemma}
\newtheorem{cor}[defn]{Corollary}
\newtheorem{prop}[defn]{Proposition}
\newtheorem{exa}[defn]{Example}
\newtheorem{rem}[defn]{Remark}

\newcommand{\N}{\mathbb{N}}

\newcommand{\F}{\mathcal{F}}
\newcommand{\T}{\mathcal{T}}

\newcommand{\lqs}{\leq^{\mathsf{q}}}
\newcommand{\lqu}{<^{\mathsf{q}}}
\newcommand{\sid}[2]{\mathsf{idx}_{ #1}\{ #2 \}}

\newcommand{\qd}[1]{\mathsf{preGQ}(#1)}
\newcommand{\lt}[1]{\mathsf{GQ}(#1 )}

\newcommand{\hra}{\hookrightarrow}

\newcommand{\C}{\mathcal{C}}

\newcommand{\uhr}{\upharpoonright}

\title{On Quasi Ordinal Diagram Systems}
\author{Mitsuhiro Okada
\institute{Department of Philosophy\\Keio University\\ Tokyo, Japan\thanks{The first author is supported by KAKENHI 17H02263, 17H02265 and JSPS-AYAME.}}
\email{mitsu@abelard.flet.keio.ac.jp}
\and
Yuta Takahashi
\institute{Department of Social Informatics\\Nagoya University\\ Nagoya, Japan\thanks{The second author is supported by KAKENHI (Grant-in-Aid for JSPS Fellows) 16J04925.}}
\institute{Japan Society for the Promotion of Science\\ Tokyo, Japan}
\email{yuuta.taka84@gmail.com}
}
\def\titlerunning{On Quasi Ordinal Diagram Systems}
\def\authorrunning{M. Okada \& Y. Takahashi}
\begin{document}
\maketitle

\begin{abstract}
The purposes of this note are the following two; we first generalize Okada-Takeuti's well quasi ordinal diagram theory, utilizing the recent result of Dershowitz-Tzameret's version of tree embedding theorem with gap conditions. Second, we discuss possible use of such strong ordinal notation systems for the purpose of a typical traditional termination proof method for term rewriting systems, especially for second-order (pattern-matching-based) rewriting systems including a rewrite-theoretic version of Buchholz's hydra game.
\end{abstract}



\section{Introduction}
Dershowitz and Tzameret extended the Friedman-K\v{r}\'{i}\v{z}-Gordeev's tree embedding theorem with gap conditions by relaxing the well orderedness condition for the labels of tree-nodes to a well quasi orderedness condition. The first purpose of our paper is to generalize Okada-Takeuti's quasi ordinal diagram systems (\cite{takeuti1987,OT1987}) using this Dershowitz-Tzameret's result. The second purpose is to analyze to which extent such a non-simplification ordering could be used as an extension of the typical termination proof method based on simplification orderings. We especially consider a ``second-order (pattern-matching-based) rewrite rule" version of Buchholz's hydra game.

A typical termination proof method for first-order term rewriting systems is to show the termination of a term rewriting system $R$ by verifying that for each rewrite rule $l (\vec{x}) \to r(\vec{x})$ of $R$, $f( l ( \vec{x}) >f( r ( \vec{x}) )$ holds, where $f$ is a strictly order-preserving mapping and $<$ is a well founded ordering with the substitution property and the monotonicity property. Here, the substitution property and the monotonicity property mean (i) for any substitution (for the list of variables) $\sigma$, if $\alpha < \beta$ holds then $\alpha \sigma < \beta \sigma$ holds, and (ii) for any context $u[\ast ]$, if $\alpha < \beta$ holds then $u[\alpha] < u[\beta]$ holds, respectively. The properties (i) and (ii) guarantee the termination of the whole $R$ because any application of (first-order) rewrite rule $l (\vec{x}) \to r(\vec{x})$ has a form $u[ l \sigma ] \to u [ r \sigma ]$ for some context $u [\ast]$ and some substitution $\sigma$.
In this note, we restrict our attention to the identity $f$ for our basic argument.

The method has been widely used for termination proofs as well as a tool for Knuth-Bendix completion. The method itself would be attractive not only for the traditional first-order rewriting but also for higher-order or graphic-pattern-matching-based rewriting. One could expect that strong and general ordering structures in proof theory would be useful for this termination proof method of higher-order pattern-matching-based rewrite systems.

However, the use of strong orderings such as $<_i$ on Takeuti's ordinal diagram systems, which is a non-simplification ordering, cannot satisfy the two basic properties (i) and (ii). Because of this difficulty, instead of the traditional termination proof method, various different techniques for the termination of higher-order rewriting systems have been utilized; for example, Jouannaud and Okada (\cite{JO1991}) introduced a generalized form of Tait-Girard's reducibility candidates method (cf. also \cite{BJO2002,BJO2018} with Blanqui).

Hence at a first look, it seems hard to adapt ordinal diagram systems to the traditional termination method. It is a natural question how we could adapt them to the termination proof method especially for higher-order rewriting systems. We aim to answer this question in the present paper.

This paper is structured as follows. We first define our generalized quasi ordinal diagram systems $\lt{I,A}$ (\S \ref{elements}), then prove the well quasi orderedness of these systems as a corollary of the Dershowitz-Tzameret's version of tree embedding theorem (\S \ref{DT}). Next, we propose a termination proof method induced by the monotonicity property and a restricted substitution property of $\lt{I,A}$ (\S \ref{app}). Finally, we take a version of Buchholz game as an example of pattern-matching-based second-order rewrite systems and show its termination by another termination proof method in terms of $\lt{I,A}$ (\S \ref{buch}). Note that the termination of the original version and its variants of Kirby-Paris's hydra game could be proved in the traditional termination method of simplification orderings (cf. \cite{isihara2007}).



\section{Well quasi ordered systems of generalized quasi ordinal diagrams}\label{sectwo}
In this section, we first generalize quasi ordinal diagram systems $(\mathsf{Q} (I,A), \leq_i )$ of Okada-Takeuti (\cite{OT1987,takeuti1987}) by using an arbitrary well partial ordering $I$ as the inner node labels (\S \ref{elements}).
Next, we prove the well quasi orderedness of these generalized systems as a direct corollary of the Dershowitz-Tzameret's version (\cite{DT2003}) of tree embedding theorem with gap condition (\S \ref{DT}).



\subsection{Formulation of generalized quasi ordinal diagram systems}\label{elements}
A \textit{quasi ordering} is a pair $(D ,\leq)$ of a set $D$ and a binary relation $\leq$ such that for any $a \in D$, $a \leq a$ holds (reflexivity) and for any $a,b,c \in D$, if $a \leq b$ and $b \leq c$ hold then $a \leq c$ holds (transitivity). A \textit{partial ordering} is a quasi ordering $(D,\leq)$ with the antisymmetry: For any $a,b \in D$, if $a \leq b$ and $b \leq a$ hold then $a = b$ holds. A \textit{linear ordering} is a partial ordering $(D,\leq)$ with the linearity: For any $a,b \in D$, $a \leq b$ or $b \leq a$ holds. A \textit{well quasi ordering} is a quasi ordering $(D, \leq )$ such that for any infinite sequence $a_0, a_1,\ldots$ from $D$, there are numbers $n$ and $m$ such that $n < m$ and $a_n \leq a_m$ holds. A \textit{well partial ordering} is a partial ordering that is a well quasi ordering. For a quasi ordering $(D,\leq )$, we use abbreviation ``$a < b$" for ``$a \leq b$ and $b \not\leq a$." Note that the well quasi orderedness has a weaker definition saying that for any infinite $\leq$-decreasing sequence $a_0 \geq a_1 \geq \ldots$ from $D$, there are numbers $n$ and $m$ such that $n < m$ and $a_n \leq a_m$ holds. A \textit{weak well quasi ordering} is a quasi ordering $(D, \leq)$ satisfying this condition. In this paper, we show the stronger version of well quasi orderedness of our quasi ordinal diagram systems.

Let $(I ,\leq_I )$ be a well partial ordering and $(A, \leq_A)$ be a well quasi ordering. We define the set $\F_0$ of constants and the set $\F_1$ of unary function symbols as follows: $\F_0 := A$, $\F_1 := \{ f_i \mid i \in I \}$. Let $\F$ be a signature defined as the union of $\F_0$, $\F_1$ and $\{ \# \}$ with a varyadic function symbol $\#$.

The \textit{pre-domain} $\qd{I,A}$ \textit{of generalized quasi ordinal diagrams} on $(I ,\leq_I )$ and $(A ,\leq_A )$ is the set $\T (\F)$ of all terms constructed from symbols in $\F$. To follow the notation in \cite{OT1987,takeuti1987}, we denote terms of $\qd{I,A}$ by $\alpha ,\beta ,\gamma , \ldots$ and adopt the following abbreviations: $(i, \alpha ) := f_i ( \alpha )$, $\alpha_1 \# \cdots \# \alpha_n := \# (\alpha_1 ,\ldots , \alpha_n)$. We call a term that belongs to $\F_0$ or has the form of $(i, \alpha )$ a \textit{connected term} of $\qd{I,A}$, and a term that has the form $\alpha_1 \# \cdots \# \alpha_n$ an \textit{unconnected term} of $\qd{I,A}$. When an unconnected term $\alpha$ is indicated as $\alpha \equiv \alpha_1 \# \cdots \# \alpha_n$, we assume that all of $\alpha_1 ,\ldots , \alpha_n$ are connected. For a connected term $\alpha$, a term $\beta$ is a \textit{component} of $\alpha$ if and only if $\beta \equiv \alpha$ holds. For an unconnected term $\alpha \equiv \alpha_1 \# \cdots \# \alpha_n$, a term  $\beta$ is a \textit{component} of $\alpha$ if and only if for some $k$ with $1 \leq k \leq n$, $\beta \equiv \alpha_k$ holds.

\begin{defn}[Labeled Finite Trees]
\textit{Labeled finite trees} are defined as follows.
\begin{enumerate}

\item A \textit{finite tree} is a partial ordering $(T, \leq_T)$ such that $T$ is a finite set with the $\leq_T$-least element called the \textit{root}, and for any $a \in T$, the set $\{b \in T \mid b \leq_T a \}$ is linearly ordered with respect to $\leq_T$.

\item Let $(I, \leq_I)$ be a quasi ordering. A $I$-\textit{labeled finite tree} is a pair of a finite tree $(T, \leq_T)$ and a mapping $l : T \to I$.

\end{enumerate}
\end{defn}
Let $(T,\leq_T)$ be a finite tree. For any $a \in T$, if the set $\{ b \in T \mid b <_T a \}$
is non-empty, then we call its greatest element
the \textit{immediate lower node}
of $a$. We call $\leq_T$-maximal elements of $T$ \textit{leaves} of $T$. The $\leq_T$-greatest lower bound of $\{a , b\}$ is denoted by $a \land b$.

\begin{exa}\label{exazero}
Let $I$ and $A$ be the following well partial ordering and well quasi ordering, respectively, where the arrow $i \rightarrow j$ means $i <_I j$ and $a \rightsquigarrow b$ means $a \leq_A b$.
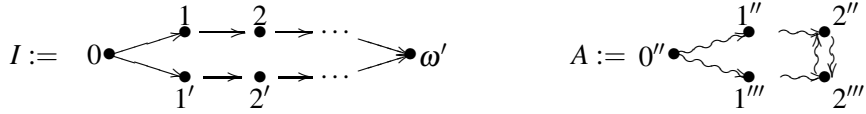
\[
\begin{xy}
(-10,0) *{I :=} ="C",
(0, 0) *{\bullet}*+!R{0} ="A",
(10, 3) *{\bullet}*+!D{1} ="A1", (20, 3) *{\bullet}*+!D{2} ="A2", (30, 3) *+{\cdots} ="A3",
(10, -3) *{\bullet}*+!U{1'} ="B1", (20, -3) *{\bullet}*+!U{2'} ="B2", (30, -3) *+{\cdots} ="B3",
(40, 0) *{\bullet}*+!L{\omega '} ="B",
\ar "A";"A1" \ar "A1";"A2" \ar "A2";"A3" \ar "A3";"B"
\ar "A";"B1" \ar "B1";"B2" \ar "B2";"B3" \ar "B3";"B"
\end{xy}
\qquad\qquad
\begin{xy}
(-10,0) *{A :=} ="C",
(0, 0) *{\bullet}*+!R{0''} ="A",
(10, 3) *{\bullet}*+!D{1''} ="A1", (20, 3) *{\bullet}*+!DL{2''} ="A2",
(10, -3) *{\bullet}*+!U{1'''} ="B1", (20, -3) *{\bullet}*+!UL{2'''} ="B2",
\ar @{~>}"A";"A1" \ar @{~>}"A1";"A2"
\ar @{~>}"A";"B1" \ar @{~>}"B1";"B2"
\ar @<1mm>@{~>}"A2";"B2" \ar @<1mm>@{~>}"B2";"A2"
\end{xy}
\]
In addition, we stipulate that for any $i \in I$ and any $a \in A$, $a <_{I} i$ hold. Here we use the symbol $\omega '$ since $I$ is similar to the ordinal $\omega$. An example of forest representation are as follows.
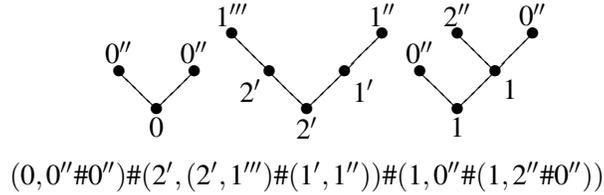
\[
\begin{xy}
(20,-9) *{(0, 0'' \# 0'') \# (2' ,(2',1''') \# (1',1'')) \# (1, 0'' \# (1, 2'' \# 0''))} ="D",
(0,0) *=0{\bullet}*+!U{0} ="A1", (-5,5) *=0{\bullet}*+!D{0''} ="A2", (5,5) *=0{\bullet}*+!D{0''} ="A3",
(20,0) *=0{\bullet}*+!U{2'} ="B1", (15,5) *=0{\bullet}*+!RU{2'} ="B2", (10,10) *=0{\bullet}*+!D{1'''} ="B3", (25,5) *=0{\bullet}*+!LU{1'} ="B4", (30,10) *=0{\bullet}*+!D{1''} ="B5",
(40,0) *=0{\bullet}*+!U{1} ="C0", (45,5) *=0{\bullet}*+!LU{1} ="C1", (40,10) *=0{\bullet}*+!D{2''} ="C2", (50,10) *=0{\bullet}*+!D{0''} ="C3", (35,5) *=0{\bullet}*+!D{0''} ="C4",
\ar @{-}"A1";"A2" \ar @{-}"A1";"A3"
\ar @{-}"B1";"B2" \ar @{-}"B2";"B3" \ar @{-}"B1";"B4" \ar @{-}"B4";"B5"
\ar @{-}"C0";"C1" \ar @{-}"C1";"C2" \ar @{-}"C1";"C3" \ar @{-}"C0";"C4"
\end{xy}
\]
\end{exa}

Note that $\#$ is represented by tree's sum and branching. By the definition of the identity below, $\#$ denotes the associative-commutative sum of connected terms, which is called ``natural sum." In term rewrite orderings such as the recursive path ordering, $\alpha_1 \# \cdots \# \alpha_n$ is often represented as $\langle \alpha_1 ,\ldots , \alpha_n \rangle$.
\begin{defn}[Identity on $\qd{I,A}$]
For any two terms $\alpha$ and $\beta$ of $\qd{I,A}$, the \textit{identity relation} $\alpha = \beta$ holds if and only if either (1) both of $\alpha$ and $\beta$ are elements of $A$, and $\alpha$ is identical with $\beta$ in the sense of $A$, or (2) $\alpha \equiv (i, \alpha ')$, $\beta \equiv (i, \beta ')$ and $\alpha ' = \beta '$ hold, or (3) $\alpha \equiv \alpha_1 \# \cdots \# \alpha_n$ and $\beta \equiv \beta_1 \# \cdots \# \beta_n$ hold with $n > 1$ and there is a permutation $p$ of $\{ 1, \ldots , n\}$ such that $\alpha_i = \beta_{p(i)}$ holds for any $i$.
\end{defn}

\begin{defn}[$i$-sections]\label{subset}
For any two terms $\alpha$ and $\beta$ of $\qd{I,A}$ and any element $i$ of $I$, the relation $\alpha \subset_i \beta$, which we call ``$\alpha$ is an $i$-\textit{section} of $\beta$", is defined as follows.
\begin{enumerate}
\item If $\beta$ is an element of $A$, then $\alpha \subset_i \beta$ never holds.
\item If $\beta \equiv (j, \beta ')$ holds, then
  \begin{enumerate}
  \item when $i = j$ holds, $\alpha \subset_i \beta$ if and only if $\alpha = \beta '$ or $\alpha \subset_i \beta '$,
  \item when $i < j$ holds, $\alpha \subset_i \beta$ if and only if $\alpha \subset_i \beta '$,
  \item when $i \not\leq j$ holds, $\alpha \subset_i \beta$ never holds.
  \end{enumerate}
\item If $\beta \equiv \beta_1 \# \ldots \# \beta_m$ holds with $m > 1$, then $\alpha \subset_i \beta$ if and only if for some $\beta_l \; (1\leq l\leq m)$, $\alpha \subset_i \beta_l$.
\end{enumerate}
\end{defn}
\begin{exa}[An example of an $i$-section]\label{exaone}
Consider the domain $\qd{I ,A}$ defined in Example \ref{exazero}.
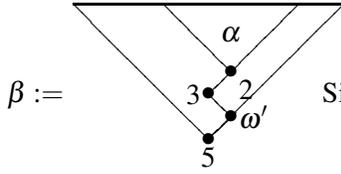
\begin{center}
\begin{xy}
(0,0) *=0{\bullet}*+!U{5} ="A", (3,3) *=0{\bullet}*+!L{\omega '} ="B", (0,6) *=0{\bullet}*+!R{3} ="C", (3,9) *=0{\bullet}*+!LU{2} ="D",
(3,14) *{\alpha} ="F", (-23,6) *{\beta :=} ="G",
(-18,18) ="E1", (-6,18) ="E2", (18,18) ="E3", (12,18) ="E4",
(70, 11) *{\text{Here $3,\omega ' ,5$ appears as labels, in this order, below $(2, \alpha )$.}} ="E5",
(70, 6) *{\text{Since $3,\omega ' ,5$ are greater than or equal to $2$, $\alpha \subset_2 \beta$ holds by definition,}} ="E6",
(70, 1) *{\text{where }\beta \equiv (5, \cdots \# (\omega ' , \cdots \# (3, \cdots \# (2, \alpha ) \# \cdots ) \# \cdots ) \# \cdots ).}
\ar @{-}"A";"B" \ar @{-}"B";"C" \ar @{-}"C";"D"
\ar @{-}"A";"E1" \ar @{-}"D";"E2" \ar @{-}"A";"E3" \ar @{-}"D";"E4" \ar @{-}"E1";"E3"
\end{xy}
\end{center}
\end{exa}
An element $i$ of $I$ is an \textit{index} of $\alpha$ if and only if there is a $\beta$ such that $\beta$ is an $i$-section of $\alpha$.
Set $\tilde{I} := I \cup \{ \infty \}$. For any $i \in I$ and any finite set $\{ \alpha_0 ,\ldots , \alpha_n \}$ of terms, define $\sid{i}{\alpha_0 ,\ldots , \alpha_n} := \{ j \mid i < j,\;j \text{ is an index of } \alpha_0 \text{ or } \ldots \text{ or } \alpha_n\}$. We denote the cardinality of $\sid{i}{\alpha_0 ,\ldots , \alpha_n}$ by $\# \sid{i}{\alpha_0 ,\ldots , \alpha_n}$. For any finite set $\{ \alpha_0 ,\ldots , \alpha_n \}$ of terms, the total number of all occurrences of $( , )$ and $\#$ in $\alpha_0 ,\ldots , \alpha_n$ is denoted by $l(\alpha_0 ,\ldots , \alpha_n)$.

\begin{defn}[Ordering $\lqs_i$ on $\qd{I,A}$]\label{less}
For any element $i$ of $\tilde{I}$, the relation $\lqs_i$ on $\qd{I,A}$ is defined by double induction on (1) $l(\alpha , \beta )$ and (2) $\# \sid{i}{\alpha , \beta}$, in this order.
\begin{enumerate}

\item If both of $\alpha$ and $\beta$ are elements of $A$, then for any element $i$ of $\tilde{I}$, $\alpha \lqs_i \beta$ if and only if $\alpha \leq_A \beta$.

\item If $\alpha$ is an element of $A$ and $\beta$ is not, then for any element $i$ of $\tilde{I}$, both $\alpha \lqs_i \beta$ and $\beta \not\lqs_i \alpha$ hold.

\item If $\alpha \equiv \alpha_1 \# \ldots \# \alpha_n$ and $\beta \equiv \beta_1 \# \ldots \# \beta_m$ hold with $n+m > 2$, then for any element $i$ of $\tilde{I}$, $\alpha \lqs_i \beta$ if and only if one of the following conditions holds:
  \begin{enumerate}
  \item there is a $\beta_l \; (1\leq l \leq m)$ such that for any $k\;(1\leq k \leq n)$, $\alpha_k \lqs_i \beta_l$ and $\beta_l \not\lqs_i \alpha_k$ hold,
  \item there is a $\beta_l \; (1\leq l \leq m)$ such that $\alpha_1 \lqs_i \beta_l$, and if $n \geq 2$ then the following holds.
    \begin{center}
    $\alpha_2 \# \ldots \# \alpha_{n} \lqs_i \beta_1 \# \ldots \# \beta_{l-1} \# \beta_{l+1} \# \ldots \# \beta_{m}$
    \end{center}
  \end{enumerate}

\item If $\alpha \equiv (j ,\alpha_0 )$, $\beta \equiv (j', \beta_0 )$ and $i = \infty$ hold, then $\alpha \lqs_{\infty} \beta$ if and only if either $j < j'$ holds or both $j = j'$ and $\alpha_0 \lqs_j \beta_0$ hold.

\item If $\alpha \equiv (j ,\alpha_0 )$ and $\beta \equiv (j', \beta_0 )$ hold and $i$ is an element of $I$, then $\alpha \lqs_i \beta$ if and only if either
\begin{enumerate}

\item[($\exists$)] there is a $\beta ' \subset_i \beta$ such that $\alpha \lqs_i \beta '$, or

\item[($\forall$)] for any $\alpha ' \subset_i \alpha$, both of $\alpha ' \lqs_i \beta$ and $\beta \not\lqs_i \alpha '$ hold, and if $\sid{i}{\alpha ,\beta} \neq \emptyset$ holds then $\alpha \lqs_{j} \beta$ holds for any minimal element $j$ of $\sid{i}{\alpha ,\beta}$, otherwise $\alpha \lqs_{\infty} \beta$ holds.

\end{enumerate}
We say $\alpha \lqs_i \beta$ \textit{holds by} $\exists$-\textit{condition}, when the condition ($\exists$) above holds. Similarly, we say $\alpha \lqs_i \beta$ \textit{holds by} $\forall$-\textit{condition}, when the condition ($\forall$) above holds. Note that if $I$ is a \textit{finite} set, $j$ in the condition ($\forall$) can be taken as a $\leq_I$-maximal element of $\{ h \mid h >_I i \}$, and the resulting ordering $\lqs_i$ becomes the same as the one defined above.
\end{enumerate}
\end{defn}

\begin{rem}
Consider $\qd{I ,A}$ defined in Example \ref{exazero}. For the readers familiar to the recursive path ordering $\geq_{rpo}$ (cf. \cite{dershowitz1987}), the table below suggests some similarity of $\geq_{i}^{\mathsf{q}}$ to $\geq_{rpo}$ and the richness of $\geq_{i}^{\mathsf{q}}$ in the sense that $\geq_{i}^{\mathsf{q}}$ for some $i \in I$ depends on $\geq_{j}^{\mathsf{q}}$ for another $j \in I$. When we stipulate that $i > \#$ holds for any $i \in I$, we have the following table:
\begin{center}
\begin{tabular}{ | l | l | }
\hline
Recursive path ordering $\geq_{rpo}$ & The ordering $\geq_{\omega '}^{\mathsf{q}}$ \\
\hline
$\alpha \equiv (i, \alpha_1 \# \cdots \# \alpha_m ) \geq_{rpo} (j, \beta_1 \# \cdots \# \beta_n ) \equiv \beta $ & $\alpha \equiv (i, \alpha_1 \# \cdots \# \alpha_m ) \geq_{\omega '}^{\mathsf{q}} (j, \beta_1 \# \cdots \# \beta_n ) \equiv \beta$ \\
if $\alpha_k \geq_{rpo} \beta$ for some $k$ with $1 \leq k \leq m$, & if $\gamma \geq_{\omega '}^{\mathsf{q}} \beta$ for some $\gamma \subset_{\omega '} \alpha$, \\
or $i > j$ and $\alpha >_{rpo} \beta_l$ for any $k$ with $1 \leq l \leq n$, & or $i > j$ and $\alpha \geq_{\omega '}^{\mathsf{q}} \delta ,\; \delta \ngeq_{\omega '}^{\mathsf{q}} \alpha$ for any $\delta \subset_{\omega '} \beta$, \\
or $i = j$ and $\alpha_1 \# \cdots \# \alpha_m \geq_{rpo} \beta_1 \# \cdots \# \beta_n$. & or $i=j$ and $\alpha \geq_{\omega '}^{\mathsf{q}} \delta ,\; \delta \ngeq_{\omega '}^{\mathsf{q}} \alpha$ for any $\delta \subset_{\omega '} \beta$ and \\
 & $\quad \alpha_1 \# \cdots \# \alpha_m \geq_{i}^{\mathsf{q}} \beta_1 \# \cdots \# \beta_n$. \\
\hline
\end{tabular}
\end{center}
\end{rem}

\begin{lem}\label{qo}
For any well partial ordering $I$, any well quasi ordering $A$ and any element $i$ of $\tilde{I}$,

\noindent$(\qd{I,A} ,\lqs_i )$ is a quasi ordering.
\end{lem}
\begin{proof}
Prove the following sublemmas (1) and (2) by double induction on $\langle l (\alpha ,\beta ,\gamma),  \# \sid{i}{\alpha ,\beta ,\gamma} \rangle$ and double induction on $\langle l(\alpha ) , \#\sid{i}{\alpha } \rangle$, respectively. (1) (Transitivity) For every $\alpha ,\beta ,\gamma \in \qd{I,A}$ and every $i \in \tilde{I}$, if $\alpha \lqs_i \beta$ and $\beta \lqs_i \gamma$ hold then $\alpha \lqs_i \gamma$ holds. (2) (Reflexivity) For every $\alpha \in \qd{I,A}$ and every $i \in \tilde{I}$, the following two hold: $\alpha \lqs_{i} \alpha$ holds, and if $i \in I$ holds then for any $\beta \in \qd{I,A}$ with $\alpha \subset_i \beta$, $\alpha \lqs_{i} \beta$ and $\beta \not\lqs_{i} \alpha$ hold.
\end{proof}

\begin{rem}
Let us comment on the definitions of $i$-sections (Definition \ref{subset}) and $\lqs_i$ (Definition \ref{less}).
\begin{itemize}

\item If we replace ``$i \not\leq j$" in Definition \ref{subset}.2.(c) with ``$i > j$", then we obtain the definition of $i$-sections in \cite{OT1987,takeuti1987}.

\item We defined $\qd{I,A}$ by taking $I$ as not a well \textit{quasi} ordering but a well \textit{partial} ordering. The antisymmetry of $I$ is needed to verify the transitivity of $\lqs_{\infty}$. In addition, note that the proof for the lemma above neither depends on the well quasi orderedness of $I$ nor the one of $A$. If we take $I$ as a well ordering, then in Definition \ref{less}.5.($\forall$) there always exists the \textit{least} element of $\sid{i}{\alpha ,\beta}$ if it is non-empty. This is the only difference of our orderings $\lqs_i$ from the ones in \cite{OT1987,takeuti1987}.

\end{itemize}
\end{rem}

The central quasi ordering $\lll$ in this paper is a generalization of the linear ordering appeared in \cite{okada1988}, which satisfies the monotonicity property. We will show in Section \ref{app} that our ordering $\lll$ satisfies the monotonicity property as well (cf. Lemma \ref{mono}).
\begin{defn}[Ordering $\alpha \lll \beta$]
For any two terms $\alpha , \beta$ of $\qd{I,A}$, $\alpha \lll \beta$ holds if and only if $\alpha \lqu_i \beta$ holds for any element $i$ of $\tilde{I}$.
\end{defn}
Let $\lll^=$ be the reflexive closure of $\lll$. It is obvious that $(\qd{I,A} ,\lll^{=})$ is a quasi ordering and that for any $\alpha$ and $\beta$, $\alpha \lll \beta$ holds if and only if both $\alpha \lll^= \beta$ and $\beta \not\lll^= \alpha$ holds.

Next, we restrict the pre-domains $\qd{I,A}$ to the \textit{path comparable tree}-domain $\lt{I,A}$, since there is a counterexample for the well quasi orderedness of $(\qd{I,A} ,\lll^= )$. In fact, we have a counterexample for its \textit{weak} well quasi orderedness.
\begin{exa}
A counterexample for the \textit{weak} well quasi orderedness of $(\qd{I,A}, \lll^= )$ is as follows: Let $I$ be $\{ 0, a_1 , a_2 , b_1, b_2\}$ and $A$ be $\{ 0 \}$, where $0 <_I a_1 <_I b_1$, $0 <_I a_2 <_I b_2$ and for any $c \in \{ a_1 ,b_1\}$ and $d \in \{ a_2 ,b_2 \}$, neither $c \leq_I d$ nor $d \leq_I c$ holds.
Then, we have the following infinite sequence.
\begin{center}
$(a_1 ,(b_2 ,0 )) \ggg ( a_1 , ( a_2 , (b_1 ,0 ))) \ggg (a_1 , ( a_2 , ( a_1 , ( b_2 , 0 )))) \ggg ( a_1 , (a_2 , (a_1 , (a_2 ,(b_1 , 0 ))))) \ggg \ldots$
\end{center}
\end{exa}

This kind of counterexamples is blocked if we restrict the pre-domains $\qd{I,A}$ to the path comparable tree-domains $\lt{I,A}$.
In the rest of this subsection and the next subsection, we assume that pre-domains $\qd{I,A}$ satisfy the following conditions: $I \cap A = \emptyset$ and $a <_{I} i$ holds for any $i \in I$ and any $a \in A$.
In addition, we identify a connected term $\alpha$ of $\qd{I,A}$ with a labeled finite tree $((T, \leq_T), l_T)$ in the manner of Example \ref{exazero}. For any quasi ordering $(A ,\leq_A)$ and any two elements $a_1 ,a_2$ of $A$, we say $a_1$ is \textit{comparable with} $a_2$ if and only if $a_1 \leq_A a_2$ or $a_2 \leq_A a_1$ holds.
\begin{defn}[Path comparable tree-domains $\lt{I,A}$ of $\qd{I,A}$, cf. \cite{DT2003}]
The \textit{path comparable tree-domain} $\lt{I,A}$ of $\qd{I,A}$ is defined as follows:
\begin{enumerate}

\item If $a \in A$ holds, then $a$ is a connected term of $\lt{I,A}$.

\item Let $\alpha_1 ,\ldots ,\alpha_n$ be connected terms of $\lt{I,A}$ with $\alpha_1 = ((T_1 ,\leq_1),l_1),\ldots ,\alpha_n = ((T_n ,\leq_n),l_n)$. If $i$ is comparable with $l_k (a)$ for any $k$ $(1 \leq k \leq n)$ and any $a \in T_k$, then $(i ,\alpha_1 \# \cdots \# \alpha_n)$ is a connected term of $\lt{I,A}$.

\item If $\alpha_1 ,\ldots ,\alpha_n$ are connected terms of $\lt{I,A}$, then $\alpha_1 \# \ldots \# \alpha_n$ is an unconnected term of $\lt{I,A}$.

\end{enumerate}

\end{defn}

A \textit{generalized quasi ordinal diagram system} is a pair $(\lt{I,A} , \lll^= )$ of a path comparable tree-domain $\lt{I,A}$ and the quasi ordering $\lll^=$ on this domain. We often abbreviate $(\lt{I,A} , \lll^= )$ as $\lt{I,A}$. Hereafter, we call a connected term (resp. an unconnected term) of $\lt{I,A}$ a \textit{connected gqod} (resp. an \textit{unconnected gqod}). For $\lt{I,A}$, we have the following lemma by induction on $l (\alpha )$.
\begin{lem}\label{cpr}
For any element $(i, \alpha)$ of $\lt{I,A}$, $\alpha \lqu_j (i, \alpha)$ holds for any element $j$ of $I$ with $j \leq i$.
\end{lem}

\subsection{Well quasi ordering proof for $(\lt{I,A}, \lll^= )$ via Dershowitz and Tzameret's tree embedding theorem}\label{DT}
We first recall Dershowitz-Tzameret's tree embedding with gap condition (cf. \cite[Definition 3.3]{DT2003}).
\begin{defn}[Tree Embedding $\hra$]
For any two connected gqod's $\alpha = ((T_1,\leq_1),l_1)$ and

\noindent $\beta = ((T_2,\leq_2),l_2)$, $\alpha \hra \beta$ holds if and only if there is an injection $\iota :T_1 \to T_2$ such that
\begin{enumerate}

\item (Node condition 1) for any $a \in T_1$, $l_1 (a) \leq_{I} l_2 (\iota (a) )$ holds,

\item (Node condition 2) for any $a_1,a_2 \in T_1$, $\iota (a_1 \wedge a_2 ) = \iota (a_1) \wedge \iota (a_2)$ holds,


\item (Edge condition) if $a_1$ is an element of $T_1$ with its immediate lower node $a_2$ in $T_1$, then for any $b \in T_2$ with $\iota (a_2) <_2 b <_2 \iota (a_1)$, $l_1 (a_1) \leq_I l_2 (b)$ holds, and   

\item (Root condition) if $a$ is the root of $T_1$, then for any $b \in T_2$ with $b <_2 \iota (a)$, $l_1 (a) \leq_I l_2 (b)$ holds.

\end{enumerate}
\end{defn}
\begin{exa}
The following is an example of Dershowitz-Tzameret's tree embedding (cf. \cite[p.87]{DT2003}).
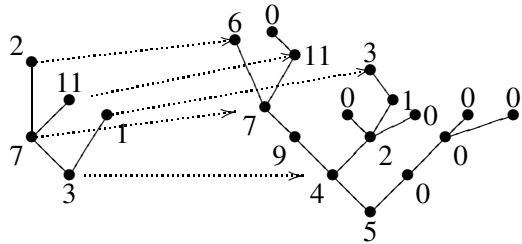
\[
\begin{xy}
(0,5) *=0{\bullet}*+!U{3} ="A1", (5,13) *=0{\bullet}*+!LU{1} ="A2", (-5,10) *=0{\bullet}*+!RU{7} ="A3", (0,15) *=0{\bullet}*+!D{11} ="A4", (-5,20) *=0{\bullet}*+!RD{2} ="A5",
(40,0) *=0{\bullet}*+!U{5} ="B1", (35,5) *=0{\bullet}*+!RU{4} ="B2", (45,5) *=0{\bullet}*+!LU{0} ="B3", (30,10) *=0{\bullet}*+!RU{9} ="B4", (40,10) *=0{\bullet}*+!LU{2} ="B5", (50,10) *=0{\bullet}*+!LU{0} ="B6", (37,13) *=0{\bullet}*+!D{0} ="B51", (43,15) *=0{\bullet}*+!L{1} ="B52", (46,13) *=0{\bullet}*+!L{0} ="B53", (40,19) *=0{\bullet}*+!D{3} ="B54", (53,13) *=0{\bullet}*+!D{0} ="B61", (59,13) *=0{\bullet}*+!D{0} ="B62", (26,14) *=0{\bullet}*+!RU{7} ="B41", (22,23) *=0{\bullet}*+!D{6} ="B42", (30,21) *=0{\bullet}*+!L{11} ="B43", (27,24) *=0{\bullet}*+!D{0} ="B44",
\ar @{-}"A1";"A2" \ar @{-}"A1";"A3" \ar @{-}"A3";"A4" \ar @{-}"A3";"A5"
\ar @{-}"B1";"B2" \ar @{-}"B1";"B3" \ar @{-}"B2";"B4" \ar @{-}"B2";"B5" \ar @{-}"B3";"B6" \ar @{-}"B5";"B51" \ar @{-}"B5";"B52" \ar @{-}"B5";"B53" \ar @{-}"B52";"B54" \ar @{-}"B6";"B61" \ar @{-}"B6";"B62" \ar @{-}"B4";"B41" \ar @{-}"B41";"B42" \ar @{-}"B41";"B43" \ar @{-}"B43";"B44"
\ar @{.>}"A1";"B2" \ar @{.>}"A2";"B54" \ar @{.>}"A3";"B41" \ar @{.>}"A4";"B43" \ar @{.>}"A5";"B42"
\end{xy}
\]
\end{exa}

\begin{thm}[Theorem 3.1 in \cite{DT2003}]\label{dt}
Let $\mathsf{GQ}^{ctd}(I,A)$ be the set of all connected gqod's from $\lt{I,A}$. Then, $(\mathsf{GQ}^{ctd}(I,A), \hra)$ is a well quasi ordering.
\end{thm}

Next, we extend this tree embedding to forests, namely, unconnected gqod's and show that $\alpha \lll^= \beta$ holds whenever $\alpha$ is embedded into $\beta$ in the sense of this forest embedding.
\begin{defn}[Forest Embedding $\hra^{\#}$]
For any two terms $\alpha$ and $\beta$ of $\lt{I,A}$ with $\alpha \equiv \alpha_1 \# \cdots \# \alpha_n$ and $\beta \equiv \beta_1 \# \cdots \# \beta_m$ ($n ,m >0$), $\alpha \hra^{\#} \beta$ holds if and only if $n \leq m$ and there is a permutation $p$ of $\{ 1,\ldots , n\}$ such that $\alpha_i \hra \beta_{p(i)}$ for any $i$ with $1 \leq i \leq n$.
\end{defn}

\begin{prop}\label{embed}
For every $\alpha ,\beta \in \lt{I,A}$, if $\alpha \hra^{\#} \beta$ holds, then $\alpha \lll^= \beta$ holds.
\end{prop}
\begin{proof}
By Lemma \ref{cpr}, one can prove the following sublemma (1) and (2). For every $\alpha ,\beta \in \lt{I,A}$ and every $i,j \in I$ with $i \leq j$, (1) if $\beta$ is connected and $(i,\alpha ) \lll^= \beta$ holds, then  $(i,\alpha ) \lll (j ,\beta)$ holds, and (2) if $\beta \equiv \beta_1 \# \cdots \# \beta_n \; (n>1)$ and for some $m \; (1 \leq m \leq n)$, $(i,\alpha ) \lll^= \beta_m$ hold, then $(i,\alpha ) \lll (j, \beta) $ holds.

We prove Proposition \ref{embed} by induction on $l (\alpha ) $. The base case is obvious. If $\alpha$ is unconnected, then the proposition immediately follows from IH. Suppose that $\alpha \equiv (i_1, \alpha ')$ holds. By the definition of $\hra^{\#}$, it suffices to prove the proposition when $\beta$ is a connected gqod $(i_2 ,\beta ')$. In this case, the proposition follows from the following two claims: For any $i \in I$ and any $\delta \subset_i \alpha$, $\delta \lqu_i \beta$ holds, and $\alpha \lqs_{\infty} \beta$ holds.

First, we show that for any $i \in I$ and any $\delta \subset_i \alpha$, $\delta \lqu_i \beta$ holds. Consider a gqod $\delta$ with $\delta \subset_i \alpha$ and put $\delta \equiv \delta_1 \# \cdots \# \delta_n \; (n\geq 1)$. We show that $\delta_k \lqu_i \beta$ for any $k$ with $1 \leq k \leq n$. If $\delta_k \in A$ holds then $\delta_k \lqu_i \beta$ obviously holds, so we assume that $\delta_k$ is of the form $(i_k , \delta_k ')$. Let us denote this outermost occurrence of $i_k$ in $\delta_k$ by $i_k^*$. Suppose that the embedding $\alpha \hra^{\#} \beta$ maps $i_k^*$ to the outermost occurrence of $j_k$ in a sub-gqod $(j_k , \beta_k)$ of $\beta$. We denote this outermost occurrence of $j_k$ in $(j_k , \beta_k)$ by $j_k^*$. Since $\delta_k \hra^{\#} (j_k , \beta_k)$ holds, we have by IH $\delta_k \lll^= (j_k , \beta_k )$ for any $k$. On the other hand, $\alpha$ has a sub-gqod $(i, \delta )$ because $\delta \subset_i \alpha$ holds. Denote this occurrence of $i$ in $\alpha$ by $i^*$ and suppose that the embedding $\alpha \hra^{\#} \beta$ maps $i^*$ to the outermost occurrence $j^*$ of $j$ in a sub-gqod $(j, \gamma_1 \# \cdots \# \gamma_m)$ of $\beta$ ($m \geq 1$), as the following figure.
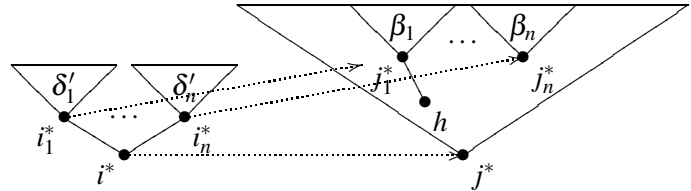
\[
\begin{xy}
(-15, 0) *=0{\bullet}*+!RU{i^*} ="A",
(-23, 5) *=0{\bullet}*+!RU{i^*_1} ="A1", (-30, 12) ="A11", (-23, 9) *{\delta_1 '} ="A12", (-16, 12) ="A13",
(-15, 5) *{\cdots} ="A2",
(-7, 5) *=0{\bullet}*+!LU{i^*_n} ="A3", (-14, 12) ="A31", (-7, 9) *{\delta_n '} ="A32", (0, 12) ="A33",
(30, 0) *=0{\bullet}*+!LU{j^*} ="B", (0, 20) ="B1", (60, 20) ="B3",
(22, 13) *=0{\bullet}*+!RU{j_1^*} ="C1", (15, 20) ="C11", (29, 20) ="C12", (22, 17) *{\beta_1} ="C13",
(25, 7) *=0{\bullet}*+!LU{h} ="C2",
(30, 15) *{\cdots} ="CD",
(38, 13) *=0{\bullet}*+!LU{j_n^*} ="D1", (31, 20) ="D11", (45, 20) ="D12", (38, 17) *{\beta_n} ="D13",
\ar@{-} "A";"A1" \ar@{-} "A";"A3"
\ar@{-} "A1";"A11" \ar@{-} "A1";"A13" \ar@{-} "A11";"A13"
\ar@{-} "A3";"A31" \ar@{-} "A3";"A33" \ar@{-} "A31";"A33"
\ar@{-} "B";"B1" \ar@{-} "B";"B3" \ar@{-} "B1";"B3"
\ar@{-} "C1";"C11" \ar@{-} "C1";"C12" \ar@{-} "C1";"C2"
\ar@{-} "D1";"D11" \ar@{-} "D1";"D12"
\ar@{.>} "A";"B" \ar@{.>} "A1";"C1" \ar@{.>} "A3";"D1"
\end{xy}
\]
Because of the gap condition of $\hra^{\#}$, we have $\delta_k \lll \gamma_1 \# \cdots \# \gamma_m$ for any $k$ by using sublemmas (1) and (2) above repeatedly. For example, if $(j_1^* ,\beta_1)$ is subsumed by $(h, \eta_1 \# (j_1^* ,\beta_1) \# \eta_2 )$ and $h \neq j^*$, then $h \geq i_1$ holds by the gap condition. Therefore, we have $\delta_1 \lll (h, \eta_1 \# (j_1^* ,\beta_1) \# \eta_2 )$ by the sublemma (2). By repeating this argument, we have $\delta_k \lll \gamma_l$ for some $l$ with $1 \leq l \leq m$. Since $\gamma_1 \# \cdots \# \gamma_m \lqu_i (j,\gamma_1 \# \cdots \# \gamma_m)$ by Lemma \ref{cpr}, $\delta_k \lqu_i (j,\gamma_1 \# \cdots \# \gamma_m)$ holds for any $k$. Therefore, $\delta \lqu_i (j,\gamma_1 \# \cdots \# \gamma_m)$ holds. Note that $i \leq h$ holds for any $h$ occurring in the path from $j^*$ to the root of $\beta$ because of the gap condition and the fact that $\delta$ is an $i$-section of $\alpha$. Therefore, by Lemma \ref{cpr}, we have $\delta \lqu_i \beta$.

By an argument similar to the one above, one can verify that $\alpha \lqs_{\infty} \beta$ holds.
\end{proof}

\begin{thm}\label{wqohra}
The system $(\lt{I,A}, \hra^{\#})$ is a well quasi ordering.
\end{thm}
\begin{proof}
Let $\{ \alpha_i \}$ and $\{ \beta_i \}$ be infinite sequences of gqod's. We say $\{ \alpha_i \}$ is a \textit{component subsequence} of $\{ \beta_i \}$ if and only if there is a monotone function $f: \N \to \N$ such that for any $i \in \N$, $\alpha_i$ is a component of $\beta_{f(i)}$. In addition, an infinite sequence $\{ a_i \}$ from a quasi ordering $(A, \leq)$ is $\leq$-\textit{bad} if and only if there is no pair $i,j$ of natural numbers such that both $i < j$ and $a_i \leq a_j$ hold.

In a way similar to the proof of \cite[Lemma 17]{OT1987}, one can prove the following claim. Let $\{ \alpha_i \}$ be an infinite sequence from $\lt{I,A}$. If there is no $\hra$-bad component subsequence $\{ \beta_i \}$ of $\{ \alpha_i \}$, then $\{ \alpha_i \}$ is $\hra^{\#}$-good, that is, $\{ \alpha_i \}$ is not $\hra^{\#}$-bad.
By this claim and Theorem \ref{dt}, we have the theorem.
\end{proof}

\begin{cor}
The system $(\lt{I,A} , \lll^{=} )$ is a well quasi ordering.
\end{cor}





\section{Discussion of generalized quasi ordinal diagram systems from the viewpoint of termination proof methods}\label{game}
In this section, we apply the well quasi orderedness of $(\lt{I,A}, \lll^=)$ to termination proof methods for second-order pattern-matching-based rewriting systems. First, we propose a termination proof method induced by the monotonicity property and a restricted substitution property of $(\lt{I,A}, \lll^=)$ (\S \ref{app}). Next, we consider Buchholz-style hydra game as an example (cf. \cite{buchholz1987,HO1998}), and show the termination of a rewriting system in this game by another termination proof method in terms of $(\lt{I,A}, \lll^=)$ (\S \ref{buch}).



\subsection{Application of $(\lt{I,A}, \lll^{=})$ to the termination proof method}\label{app}
Let $V$ be a finite set of variables. In this and next subsections, we consider an arbitrary system $\lt{I,A}$ satisfying the following conditions: (1) $A = V \cup \{ \rho \}$ holds, and neither $x \leq_A y$ nor $y \leq_A x$ holds for any $x,y \in V$, and $\rho$ is the $\leq_A$-minimum element of $A$, and (2) $\rho \in I$ and $I \cap V = \emptyset$ hold, and $\rho <_I x <_I i$ holds for any $x \in V$ and any $i \in I \setminus \{ \rho \}$.
We denote the set of all contexts from $\lt{I,A}$ by $\C$.

A \textit{numeral substitution} $\sigma$ for a gqod $\alpha$ is a substitution assigning a numeral term to each variable in $\alpha$, where a \textit{numeral term} is a connected gqod that consists of $\rho$'s only. Note that the numeral terms are generalization of the numerals represented by no $\#$-branching trees $\rho - \rho - \cdots - \rho$; the numeral terms correspond to the ordinals up to $\varepsilon_0$.

As we will see below, the system $(\lt{I,A}, \lll^{=})$ satisfies not only the monotonicity property (Lemma \ref{mono}) but also the following restricted substitution property (Lemma \ref{num}): For any \textit{connected} terms $\alpha$ and $\beta$ of $\lt{I,A}$ with $\alpha ,\beta \not\in A$, if $\alpha \lll \beta$ holds then $\alpha\sigma \lll \beta\sigma$ holds for any \textit{numeral} substitution $\sigma$. Therefore, if $l \ggg r$ holds for any rule $l \rhd r$ of a rewriting system $R$, then we immediately obtain the termination of the $R$-rewrite relation $\to$ defined as follows: $\alpha \to \beta$ holds if and only if $\alpha \equiv (i, u[ l ])\sigma$ and $\beta \equiv (i, u[ r  ])\sigma$ for some context $u[\ast ]$ and some numeral substitution $\sigma$.
\begin{lem}[The monotonicity property lemma]\label{mono}
Let $\alpha$ and $\beta$ be gqod's with $u[\alpha ] \in \lt{I,A}$ and $u[\beta ] \in \lt{I,A}$. If $\alpha \lll \beta$ holds, then $u[ \alpha  ] \lll u[ \beta ]$ holds.
\end{lem}
\begin{proof}
We easily obtain $\gamma_1 \# \cdots \# \gamma_k \# \alpha \# \gamma_{k+1} \# \cdots \# \gamma_n \lll \gamma_1 \# \cdots \# \gamma_k \# \beta \# \gamma_{k+1} \# \cdots \# \gamma_n$ if $\alpha \lll \beta$ holds. Therefore, it suffices to show that for any two gqod's $\alpha$ and $\beta$ with $(i,\alpha ) \in \lt{I,A}$ and $(i,\beta ) \in \lt{I,A}$, if $\alpha \lll \beta$ holds then $(i,\alpha ) \lll (i, \beta )$ holds.

Assume that $\alpha \lll \beta$ holds, then it is obvious that $(i,\alpha ) \lqu_{\infty} (i,\beta )$ holds.
We show that $(i,\alpha ) \lqu_j (i,\beta )$ holds for any $j$ with $j \not\leq i$. In this case, neither $(i,\alpha )$ nor $(i, \beta )$ has a $j'$-section for any $j'$ with $j' \geq j$, so $(i, \alpha ) \lqs_j (i,\beta )$ holds by $\forall$-condition, since we already have $(i,\alpha ) \lqs_{\infty} (i,\beta )$. In the present case, it is easy to verify that $(i,\beta ) \not\leq^{\mathsf{q}}_{j} (i,\alpha )$ holds.

Finally, we show that $(i,\alpha ) \lqu_j (i,\beta )$ holds for any $j$ with $j \leq i$ by induction on $\# \sid{j}{(i,\alpha ), (i,\beta )}$. Consider the base case, where $\# \sid{j}{(i,\alpha ), (i,\beta )} =0$ holds, that is, $j = i$ holds. In the base case, we have $(i,\alpha ) \lqs_i (i,\beta )$ by $\forall$-condition, since $(i,\alpha ) \lqs_{\infty} (i,\beta )$ holds and for any $\gamma \subset_i (i,\alpha )$, $\gamma \lqs_i \alpha \lqu_{i} \beta \lqu_i (i, \beta )$ holds. Suppose that $(i,\beta ) \lqs_i (i,\alpha )$ holds. If $(i,\beta ) \lqs_i (i,\alpha )$ holds by $\exists$-condition, there is a gqod $\gamma \subset_i (i,\alpha )$ such that $(i,\beta ) \lqs_i \gamma \lqs_i \alpha$ holds. Then, we have $\beta \lqu_i \alpha$, contradiction. We have a contradiction as well, if $(i,\beta ) \lqs_i (i,\alpha )$ holds by $\forall$-condition. Therefore, $(i,\beta ) \not\leq^{\mathsf{q}}_{i} (i,\alpha )$ holds.

Consider the induction step. In a way similar to the base case, we have $\gamma \lqu_j (i,\beta )$ for any $\gamma \subset_j (i, \alpha )$. Note that we have $\beta \lqu_j (i, \beta)$ by Lemma \ref{cpr} because $(i,\beta) \in \lt{I,A}$ holds. Let $j'$ be a minimal element of $\sid{j}{(i,\alpha ),(i,\beta )}$. Since $j'$ is an index of $(i,\alpha )$ or $(i,\beta )$, $j' \leq i$ holds. By IH, we have $(i,\alpha ) \lqs_{j'} (i,\beta )$. If $(i,\beta ) \lqs_j (i,\alpha )$ holds then we obtain a contradiction, so $(i,\alpha ) \lqu_j (i,\beta )$ holds.
\end{proof}

As stated above, we have the following restricted substitution property. Note that we cannot have the full substitution property of $(\lt{I,A}, \lll^=)$ because there is a counterexample. Take a substitution $\sigma$ assigning $(\rho ,\rho)$ to $x$. Then, we have both $x \# x \lll (\rho ,\rho ) \# x$ and $(x \# x)\sigma \not\lll ((\rho ,\rho ) \# x)\sigma$.
\begin{lem}[The numeral substitution property lemma]\label{num}
For any connected gqod's $\alpha$ and $\beta$ of $\lt{I,A}$ with $\alpha ,\beta \not\in A$, if $\alpha \lll \beta$ holds then $\alpha\sigma \lll \beta\sigma$ holds for any numeral substitution $\sigma$.
\end{lem}
\begin{proof}
Verify the following by double induction on $\langle l (\alpha ,\beta ) , \# \sid{i}{\alpha ,\beta}\rangle$: For any $i \in \tilde{I}$, any connected terms $\alpha$ and $\beta$ with $\alpha ,\beta \not\in A$, if $\alpha \lqu_i \beta$ holds then $\alpha\sigma \lqu_i \beta\sigma$ holds for any numeral substitution $\sigma$.
\end{proof}

\begin{prop}
Let $l \rhd r$ be an  arbitrary rule of a rewriting system $R$. For any $i \in I$ and any context $u[\ast]$, if both of $(i, u[l])$ and $(i, u[r])$ are terms of $\lt{I,A}$ and $l \ggg r$ holds, then $(i, u[l])\sigma \ggg (i, u[r])\sigma$ holds for any numeral substitution $\sigma$.
\end{prop}
\begin{proof}
By Lemma \ref{mono} and Lemma \ref{num}.
\end{proof}



\subsection{A case study with the rewrite system of Buchholz-style hydra game}\label{buch}
Below we study the termination proof of a rewrite system version of Buchholz-style hydra game, by using our generalized quasi ordinal diagram systems.
Instead of the general method for termination proofs proposed in the last subsection, we use the method relative to a given rewriting system.

The set $\C_i$ of contexts for any $i \in I$ is defined as follows:
\begin{enumerate}

\item $\ast$ is a connected term of $\C_i$.

\item if $u_i [ \ast ]$ is a connected term of $\C_i$ and $\alpha_1 ,\ldots ,\alpha_n ,\beta_1 ,\ldots ,\beta_m$ are connected gqod's in $\lt{I,A}$, then $u_i' [\ast] \equiv \alpha_1 \# \cdots \# \alpha_n \# u_i [\ast] \# \beta_1 \# \cdots \# \beta_m$ is an unconnected term of $\C_i$.

\item if $u_i [\ast]$ is a term of $\C_i$ and $(j , u_i[\rho ]) \in \lt{I,A}$ holds with $j \geq i$, then $(j ,u_i [\ast] )$ is a connected term of $\C_i$.
\end{enumerate}
Note that $u[ \alpha ]$ and $u_i [\alpha ]$ may not be path comparable trees even if $u[ \ast ]$, $u_i [\ast ]$ and $\alpha$ are.

An element $i$ of $I$ is a \textit{successor} if the set of all $I$-elements smaller than $i$ has a maximal element. An element $\lambda$ of $I$ is a \textit{limit element} if the set of all $I$-elements smaller than $\lambda$ is non-empty and does not have a maximal element.

To formulate Buchholz-style hydra game, we define the notion of \textit{segments}.
\begin{defn}[Segments]
For any connected gqod $\alpha$ in $\lt{I,A}$ and any finite subset $J$ of $I$, the \textit{segment} $\alpha \uhr J$ of $\alpha$ on $J$ is defined as follows: If $J$ is empty, $\alpha \uhr J := \alpha$. Assume that $J$ is not empty. 
\begin{enumerate}

\item If $\alpha \in A$ holds, then $\alpha \uhr J := \alpha$.

\item Suppose that $\alpha$ is of the form $(i, \alpha_1 \# \cdots \# \alpha_n )$ with $\alpha_1 , \cdots , \alpha_n$ connected. If there is an element of $J$ that is not comparable with $i$ then $\alpha \uhr J := \rho$, otherwise $\alpha \uhr J := (i, \alpha_1 \uhr J \# \cdots \# \alpha_n \uhr J )$.
\end{enumerate}
\end{defn}

 For any $n \in \N$, let $\alpha \cdot (n+1)$ be $\underbrace{\alpha \# \cdots \# \alpha}_{\text{$n+1$ times}}$. We define $I \uhr i := \{ j \in I  \mid j < i \}$. A \textit{substitution} $\sigma$ is a mapping from $V$ to $\lt{I,A}$.
\begin{defn}[\textbf{Buchholz-style hydra game rules}]
Buchholz-style hydra game rules on $\lt{I,A}$ are the following rules (1), (1)', (2) and (3):
\begin{enumerate}

\item[(1)] $( i, \alpha_1 \# \cdots \# \alpha_n \# (\rho ,\rho) \# \beta_1 \# \cdots \# \beta_m ) \rhd  ( (i,  \alpha_1 \# \cdots \# \alpha_n \# \beta_1 \# \cdots \# \beta_m ) \cdot ( k+1 ) ) \# \rho \cdot 2 $, where $k$ is an arbitrarily chosen natural number and $n+m >0$ holds.

\item[(1)'] $( i, (\rho ,\rho) \bigr) \rhd  \bigl( (i, \rho ) \cdot (k+1) ) \# \rho $, where $k$ is an arbitrarily chosen natural number.

\item[(2)] $( j , \alpha_1 \# \cdots \# \alpha_n \# u_i [ (i ,a) ] \# \beta_1 \# \cdots \# \beta_m ) \rhd \Biggl( j, \alpha_1 \# \cdots \# \alpha_n \# u_i \biggl[ \Bigl( i^- , u_i [ (\rho ,a) ]  \Bigr) \uhr J \biggr] \# \beta_1 \# \cdots \# \beta_m \Biggr) \# \rho$,

where $a \in A$ and $j < i$ hold, $i$ is a successor, $i^-$ is an arbitrarily chosen maximal element of $I \uhr i$, $u_i [\ast ]$ is a connected term of $\mathcal{C}_i$ and $J := \emptyset$ if $u_i [\ast ] = \ast$, otherwise
\begin{center}
$ J := \{ j , i^- \} \cup \{ h \in I \mid \text{$h$ is a label occurring in the path between $\ast$ and the root of $u_i[\ast]$}\}$.
\end{center}

\item[(3)] $(\lambda ,a) \rhd (i ,a ) \# \rho $, where $\lambda$ is a limit element and $i$ is an arbitrarily chosen element of $I \uhr \lambda$. 

\end{enumerate}

\end{defn}

\begin{rem}\label{seg}
Let us comment on the definition above.
\begin{itemize}

\item Choices in each of the rules are, intuitively, made by hydras after Heracles's attacks.

\item By restricting the initial states of hydras to the form $(\rho , \alpha_1 ) \# \cdots \# (\rho , \alpha_n)$, one can guarantee that at least one rule applies to the initial hydras of the games. Moreover, if rewriting in these games terminates, then the final hydra is a forest of numeral terms that have at most the depth $2$.

\item The reason we attached one or two $\rho$'s to the right gqod's of the rules above is that we indicate by the number of the occurrences of $\rho$'s at leaves how many times at most we applied the rules.

\item We take the segment $\Bigl( i^- , u_i [ (\rho ,a) ]  \Bigr) \uhr J$ at RHS in the rule (2) to make sure that the right gqod in the rule (2) belongs to $\lt{I,A}$, since $u_i \Bigl[ \Bigl( i^- , u_i [ (\rho ,a) ]  \Bigr)  \Bigr]$ is not always a path comparable tree.


\item Our hydras have labels from the well partial ordering $I$, while Buchholz's original hydras (\cite{buchholz1987}) have labels from $\N \cup \{ \omega \}$. Moreover, our rule (3) is more liberal than Buchholz's corresponding rule. On the other hand, the duplication in the rule (2) of our hydra game is more restricted than the one in Buchholz's corresponding rule. Because of this restriction, we have Lemma \ref{sub} below.

\end{itemize}
\end{rem}

\begin{defn}[Rewrite relation for Buchholz-style hydra game]
For any two terms $\alpha$ and $\beta$ of $\lt{I,A}$, $\alpha \to \beta$ holds if and only if $\alpha \equiv u[l \sigma]$ and $\beta \equiv u[r \sigma]$ hold for some context $u[\ast]$, some substitution $\sigma$ and some $l ,r \in \lt{I,A}$ with $l \rhd r$.
\end{defn}
\begin{exa}
Consider $\lt{I ,V \cup \{ 0\}}$ with $I$ defined in Example \ref{exazero}. Then, a play runs as follows.
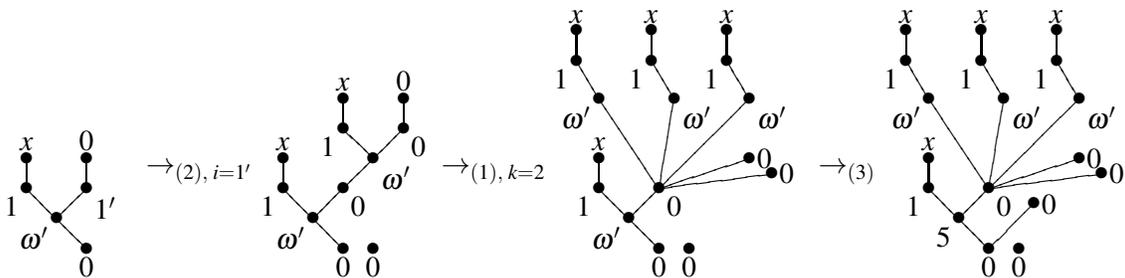
\[
\begin{xy}
(0, 0) *=0{\bullet}*+!U{0} ="A", (-4, 4) *=0{\bullet}*+!RU{\omega '} ="A0",
(-8, 8) *=0{\bullet}*+!RU{1} ="A00", (0, 8) *=0{\bullet}*+!LU{1'} ="A01",
(-8, 12) *=0{\bullet}*+!D{x} ="A000", (0, 12) *=0{\bullet}*+!D{0} ="A010",
(15, 10) *{\to_{(2),\; i= 1'}},
\ar@{-} "A";"A0" \ar@{-} "A0";"A00" \ar@{-} "A0";"A01" \ar@{-} "A00";"A000" \ar@{-} "A01";"A010"
\end{xy}
\begin{xy}
(4, 0) *=0{\bullet}*+!U{0} ="B",
(0, 0) *=0{\bullet}*+!U{0} ="A", (-4, 4) *=0{\bullet}*+!RU{\omega '} ="A0",
(-8, 8) *=0{\bullet}*+!RU{1} ="A00", (-8, 12) *=0{\bullet}*+!D{x} ="A000",
(0, 8) *=0{\bullet}*+!LU{0} ="A01", (4, 12) *=0{\bullet}*+!LU{\omega '} ="A010",
(0, 16) *=0{\bullet}*+!RU{1} ="A0100", (0, 20) *=0{\bullet}*+!D{x} ="A01000",
(8, 16) *=0{\bullet}*+!LU{0} ="A0101", (8, 20) *=0{\bullet}*+!D{0} ="A01010",
(20, 10) *{\to_{(1),\; k = 2}},
\ar@{-} "A";"A0" \ar@{-} "A0";"A00" \ar@{-} "A0";"A01" \ar@{-} "A00";"A000"
\ar@{-} "A01";"A010" \ar@{-} "A010";"A0100" \ar@{-} "A0100";"A01000"
\ar@{-} "A010";"A0101" \ar@{-} "A0101";"A01010"
\end{xy}
\begin{xy}
(4, 0) *=0{\bullet}*+!U{0} ="B",
(0, 0) *=0{\bullet}*+!U{0} ="A", (-4, 4) *=0{\bullet}*+!RU{\omega '} ="A0",
(-8, 8) *=0{\bullet}*+!RU{1} ="A00", (-8, 12) *=0{\bullet}*+!D{x} ="A000",
(0, 8) *=0{\bullet}*+!LU{0} ="A01",
(-8, 20) *=0{\bullet}*+!RU{\omega '} ="A010", (-11, 25) *=0{\bullet}*+!RU{1} ="A0100", (-11, 29) *=0{\bullet}*+!D{x} ="A01000",
(2, 20) *=0{\bullet}*+!LU{\omega '} ="A011", (-1, 25) *=0{\bullet}*+!RU{1} ="A0110", (-1, 29) *=0{\bullet}*+!D{x} ="A01100",
(12, 20) *=0{\bullet}*+!LU{\omega '} ="A012", (9, 25) *=0{\bullet}*+!RU{1} ="A0120", (9, 29) *=0{\bullet}*+!D{x} ="A01200",
(15, 10) *=0{\bullet}*+!L{0} ="A013", (12, 12) *=0{\bullet}*+!L{0} ="A014",
(25, 10) *{\to_{(3)}}
\ar@{-} "A";"A0" \ar@{-} "A0";"A00" \ar@{-} "A0";"A01" \ar@{-} "A00";"A000"
\ar@{-} "A01";"A010" \ar@{-} "A010";"A0100" \ar@{-} "A0100";"A01000"
\ar@{-} "A01";"A011" \ar@{-} "A011";"A0110" \ar@{-} "A0110";"A01100"
\ar@{-} "A01";"A012" \ar@{-} "A012";"A0120" \ar@{-} "A0120";"A01200"
\ar@{-} "A01";"A013" \ar@{-} "A01";"A014"
\end{xy}
\begin{xy}
(4, 0) *=0{\bullet}*+!U{0} ="B",
(0, 0) *=0{\bullet}*+!U{0} ="A", (-4, 4) *=0{\bullet}*+!RU{5} ="A0", (6, 6) *=0{\bullet}*+!L{0} ="A1",
(-8, 8) *=0{\bullet}*+!RU{1} ="A00", (-8, 12) *=0{\bullet}*+!D{x} ="A000",
(0, 8) *=0{\bullet}*+!LU{0} ="A01",
(-8, 20) *=0{\bullet}*+!RU{\omega '} ="A010", (-11, 25) *=0{\bullet}*+!RU{1} ="A0100", (-11, 29) *=0{\bullet}*+!D{x} ="A01000",
(2, 20) *=0{\bullet}*+!LU{\omega '} ="A011", (-1, 25) *=0{\bullet}*+!RU{1} ="A0110", (-1, 29) *=0{\bullet}*+!D{x} ="A01100",
(12, 20) *=0{\bullet}*+!LU{\omega '} ="A012", (9, 25) *=0{\bullet}*+!RU{1} ="A0120", (9, 29) *=0{\bullet}*+!D{x} ="A01200",
(15, 10) *=0{\bullet}*+!L{0} ="A013", (12, 12) *=0{\bullet}*+!L{0} ="A014",
\ar@{-} "A";"A0" \ar@{-} "A";"A1" \ar@{-} "A0";"A00" \ar@{-} "A0";"A01" \ar@{-} "A00";"A000"
\ar@{-} "A01";"A010" \ar@{-} "A010";"A0100" \ar@{-} "A0100";"A01000"
\ar@{-} "A01";"A011" \ar@{-} "A011";"A0110" \ar@{-} "A0110";"A01100"
\ar@{-} "A01";"A012" \ar@{-} "A012";"A0120" \ar@{-} "A0120";"A01200"
\ar@{-} "A01";"A013" \ar@{-} "A01";"A014"
\end{xy}
\]
\end{exa}

The termination proof method in this subsection consists in the following substitution property.
\begin{lem}[The relative substitution property lemma]\label{sub}
For any substitution $\sigma$ and any two gqod's $l$ and $r$, if both of $l\sigma$ and $r\sigma$ belong to $\lt{I,A}$ and $l \rhd r$ holds, then $l\sigma \ggg r\sigma$ holds.
\end{lem}
\begin{proof}
We consider the rule (2) only, since the cases of the other rules are obvious. For any substitution $\sigma$ satisfying the conditions of the lemma, we show
\begin{center}
$( j , \vec{\alpha \sigma} \# u_i \sigma [ (i ,a\sigma) ] \# \vec{\beta \sigma} \rhd \Bigl( j, \vec{\alpha \sigma} \# u_i \sigma \Bigl[ \Bigl(  \bigl( i^- ,  u_i [ (\rho ,a) ] \bigr) \uhr J \Bigr) \sigma \Bigr] \# \vec{\beta \sigma} \Bigr) \# \rho $,
\end{center}
where we abbreviate $\alpha_1\sigma \# \cdots \# \alpha_n\sigma$ and $\beta_1\sigma \# \cdots \# \beta_m\sigma$ as $\vec{\alpha\sigma}$ and $\vec{\beta\sigma}$, respectively. We verify the claim above by proving the following two sublemmas. The lemma follows from these sublemmas by the monotonicity. Note that one easily obtains $(\alpha \uhr J )\sigma \lll^= \alpha\sigma$ by induction on $l(\alpha)$.
%
\\[10pt]
\textbf{Sublemma 1.}
$u_i\sigma \Bigl[ \Bigl(  \bigl( i^- ,  u_i [ (\rho ,a) ] \bigr) \uhr J \Bigr) \sigma \Bigr] \lqu_h u_i \sigma [(i ,a\sigma )]$ holds for any $h$ with $h \not\leq i^-$ or $h = \infty$.

By induction on the built-up of $u_i\sigma [\ast ]$. If $u_i\sigma [\ast ]$ is $\ast$, then one can prove this sublemma by subinduction on $\#\sid{h}{ ( ( i^- ,  u_i [ (\rho ,a) ] ) \uhr J ) \sigma ,  (i ,a\sigma )}$.

Suppose that $u_i\sigma [\ast ]$ is $(k, \gamma_1 \# \cdots \# \gamma_l \# u_i '[ \ast ] \# \delta_1 \# \cdots \# \delta_{l'})$ with $u_i '[\ast ]$ connected. We abbreviate $\gamma_1 \# \cdots \# \gamma_l$ and $\delta_1 \# \cdots \# \delta_{l'}$ as $\vec{\gamma}$ and $\vec{\delta}$, respectively. In addition, we set
\begin{center}
$\alpha :\equiv (k, \vec{\gamma} \# u_i ' \Bigl[ \Bigl(  \bigl( i^- ,  u_i [ (\rho ,a) ] \bigr) \uhr J \Bigr) \sigma \Bigr] \# \vec{\delta} ),\quad \beta :\equiv  (k, \vec{\gamma} \# u_i ' [ (i, a\sigma )] \# \vec{\delta} )$.
\end{center}
By IH, we have $u_i ' \Bigl[ \Bigl(  \bigl( i^- ,  u_i [ (\rho ,a) ] \bigr) \uhr J \Bigr) \sigma \Bigr] \lqu_h u_i ' [(i ,a\sigma )]$ for any $h \not\leq i^-$. Therefore, we in particular have $u_i ' \Bigl[ \Bigl(  \bigl( i^- ,  u_i [ (\rho ,a) ] \bigr) \uhr J \Bigr) \sigma \Bigr] \lqu_k u_i ' [(i ,a\sigma )]$, so $\alpha  \lqu_{\infty} \beta$ holds. Then, it follows that $\alpha \lqu_{h} \beta$ holds for any $h \not\leq i^-$ with $h \not\leq k$. Since $\alpha ' \lqu_k \beta$ holds for any $\alpha ' \subset_k \alpha$, we also have $\alpha \lqu_k \beta$ by $\forall$-condition. Finally, consider $h \in I$ with $h \not\leq i^-$ and $h < k$. Then, for any $\alpha ' \subset_h \alpha$, we have $\alpha ' \lqu_h \beta$ also in this case. Then, by subinduction on $\#\sid{h}{\alpha ,\beta}$, it follows that $\alpha \lqu_h \beta$ holds.
\\[10pt]
\textbf{Sublemma 2.}
$u_i\sigma \Bigl[ \Bigl(  \bigl( i^- ,  u_i [ (\rho ,a) ] \bigr) \uhr J \Bigr) \sigma \Bigr] \lqu_h u_i \sigma [(i ,a\sigma )]$ holds for any $h \leq i^-$.

Take a gqod $\gamma \subset_h u_i\sigma \Bigl[ \Bigl(  \bigl( i^- ,  u_i [ (\rho ,a) ] \bigr) \uhr J \Bigr) \sigma \Bigr]$ with $h \leq i^-$ and consider the key case, where $\gamma \subset_h \Bigl(  \bigl( i^- ,  u_i [ (\rho ,a) ] \bigr) \uhr J \Bigr) \sigma$ holds. Then, there is a gqod $\alpha '$ such that $\gamma \equiv (\alpha ' \uhr J)\sigma$ holds. Since we have
\begin{center}
$(\alpha ' \uhr J)\sigma \lll \alpha ' \sigma \leq_h  (u_i [(\rho ,a )] ) \sigma = u_i \sigma [ (\rho ,a\sigma ) ] \lll u_i \sigma [(i ,a\sigma )]$,
\end{center}
$\gamma \lqu_h u_i \sigma [(i ,a\sigma )]$ holds. In the other cases, one easily see $\gamma \lqu_h u_i \sigma [(i ,a\sigma )]$. We obtain the present sublemma by induction on $\#\sid{h}{ u_i\sigma \Bigl[ \Bigl(  \bigl( i^- ,  u_i [ (\rho ,a) ] \bigr) \uhr J \Bigr) \sigma \Bigr] , u_i \sigma [(i ,a\sigma )] }$ and Sublemma 1.
\end{proof}

\begin{prop}[Termination of $\to$]
For any substitution $\sigma$, any context $u[\ast ]$ and any two gqod's $l$ and $r$, if both of $u[ l\sigma ]$ and $u [ r\sigma ]$ belong to $\lt{I,A}$ and $l \rhd r$ holds, then $u[l\sigma ] \ggg u[r\sigma ]$ holds.
\end{prop}
\begin{proof}
By Lemma \ref{sub} and Lemma \ref{mono}.
\end{proof}

\section{Concluding discussion and future work}
We have generalized Okada-Takeuti's quasi ordinal diagram systems and proved the well quasi orderedness of the generalized systems $(\lt{I,A} , \lll^=)$, using Dershowitz-Tzameret's tree embedding theorem with gap conditions. This gives one example of usefulness of Dershowitz-Tzameret's tree embedding theorem. We also have examined to which extent $(\lt{I,A} , \lll^=)$ can be used for termination proof methods for higher-order rewrite systems, by proposing two termination proof methods. First, it has been shown that $(\lt{I,A} , \lll^=)$ satisfies not only the monotonicity property but also the numeral substitution property, which holds for the substitutions of numeral trees.
Next, we have formulated the termination proof method relative to a given higher-order rewrite system, taking Buchholz-style hydra game as an example.

We conjecture that the order on $\lt{I,A}$ decreases even if Buchholz-style hydra game allows to move a subtree on \textit{several} leave nodes of another subtree at once. This might lead us to investigate an alternative substitution property, rather than the numeral substitution property reported in this paper.
These computational phenomena will be examined in our future work.
We also investigate further how to use our results of generalized quasi ordinal diagram systems for another pattern-matching-based rewrite programming. In addition, we are working on more graphic (non-tree) versions of quasi ordinal diagrams, based on the current results.

\nocite{*}
\bibliographystyle{eptcs}
\bibliography{bibeptcs2}
\end{document}